\documentclass[conference,a4paper]{IEEEtran}

\usepackage{amsmath,amssymb,amsfonts,graphicx}
\usepackage{latexsym}
\usepackage{amsthm}

\newcommand{\e}{{\epsilon}}    
\newcommand{\V}{{\mathcal{V}}}  
\newcommand{\E}{{\mathcal{E}}}  
\newcommand{\N}{{\mathcal{N}}}  
\newcommand{\W}{{\mathcal{W}}}  
\newcommand{\G}{{\mathcal{G}}}  
\newcommand{\C}{{\mathcal{C}}}  
\newcommand{\R}{{\textbf{R}}}   

\newtheorem{theorem}{Theorem}

\begin{document}
\sloppy

\title{On the Capacity Region of Broadcast Packet Erasure Relay Networks With Feedback}

\author{
  \IEEEauthorblockN{Quan Geng}
  \IEEEauthorblockA{Google Research\\
    New York, NY 10011\\
    Email: qgeng@google.com}
    \\
\IEEEauthorblockN{Mindi Yuan}
 \IEEEauthorblockA{CSL and Dept. of ECE\\
   UIUC, IL 61801\\
   Email: myuan5@illinois.edu}
 \and
   \IEEEauthorblockN{Hieu T. Do}
  \IEEEauthorblockA{ Ericsson Research, \\
	Stockholm, Sweden\\
  Email: hieu.do@ericsson.com} 
 \\
 \IEEEauthorblockN{Yun Li}
 \IEEEauthorblockA{CSL and Dept. of ECE\\
   UIUC, IL 61801\\
   Email: yunli2@illinois.edu}
   \and
 \IEEEauthorblockN{Rui Wu}
 \IEEEauthorblockA{CSL and Dept. of ECE\\
   UIUC, IL 61801\\
   Email: ruiwu1@illinois.edu}\\
 \IEEEauthorblockN{Wei Ding}
 \IEEEauthorblockA{Google Research\\
   New York, NY 10011\\
   Email: vvei@google.com}
}

\maketitle

\begin{abstract}
We derive a new outer bound on the capacity region of broadcast traffic in multiple input broadcast packet erasure channels with feedback, and extend this outer bound to packet erasure relay networks with feedback.  We show the tightness of the outer bound for various classes of networks. An important engineering implication of this work is that for network coding schemes for parallel broadcast channels, the ``xor'' packets should be sent over correlated broadcast subchannels\footnote{This paper was published in the Proceedings of  the 53rd IEEE International Conference on Communications (ICC), Shanghai, 2019.}.
\end{abstract}

\section{Introduction} \label{sec:intro}

The broadcast packet erasure channel was introduced in \cite{Dana06}, which captures the broadcast nature of wireless communication \cite{Tse2005,GV_JSAC14, GL_ISIT14}. In packet erasure relay networks, each directed link connecting two nodes $i$ and $j$ is modeled as a packet erasure channel, which is a natural generalization of the binary erasure channel from binary symbol to packet. During each time slot, node $i$ can send out a packet of fixed size to node $j$. With probability $(1-\e_{ij})$ node $j$ receives the whole packet correctly, and with probability $\e_{ij}$ the packet is erased by the channel. Furthermore, due to the broadcast nature of wireless communication, during each transmission the transmitter sends out the same packet to all the nodes it is connected to. Additionally, in packet erasure relay networks we assume that there is no interference, i.e., each node can receive packets sent from different nodes simultaneously without interference. This assumption is valid when some mechanisms in practical systems are implemented to avoid interference, e.g., via frequency-division multiplexing.
In packet erasure channels with feedback,  each transmitter can get the channel output feedback immediately from the receivers after each transmission, i.e., the transmitter will know whether each receiver has got the packet or not. In practice, this type of feedback can be obtained using the Automatic Repeat-reQuest (ARQ) mechanism.

\subsection{Existing Work}

The broadcast capacity region of multiple input broadcast packet erasure channels without feedback was characterized in \cite{DH05} and can be achieved by the time sharing scheme. It turns out that feedback can significantly improve the capacity region.  The capacity region of $1$-to-$2$ broadcast packet erasure channels with feedback was derived in \cite{GT09} and is achievable by a simple networking coding scheme. 
The general $1$-to-$K$ broadcast packet erasure channels with feedback was independently studied in \cite{Gat10} and \cite{CCWangIT2012}, which characterized the capacity region for $K=3$. 
An XOR-based encoding scheme under a similar setting for three-user broadcast erasure channels with feedback was designed in \cite{Tassiulas14}. 
The linear network coding capacity region of two-receiver MIMO broadcast packet erasure channels with feedback was derived in \cite{Wang2012} and \cite{WangIT14}. For broadcast packet erasure relay networks with feedback, a random linear network coding scheme can achieve the capacity of unicast traffic \cite{Dana06}. For unicast traffic on packet erasure networks with feedback, a simple capacity-achieving ``dynamical routing'' scheme, which can be written as a linear program, was presented in \cite{GMS10}. It was shown that  local network coding and global routing can achieve the cut-set bound within a factor of $O(\log^3 k \log d_{\max})$ for $k$-unicast traffic in broadcast packet relay networks with commensurate feedback \cite{Sreeram11}, where $d_{\max}$ is the maximum degrees of nodes in the network. One direct extension of \cite{Sreeram11} is that for broadcast traffic, the same separation scheme can also achieve the cut-set bound within a factor of $O(\log^3 k \log d_{\max})$. An align-and-forward relaying communication scheme was studied in \cite{Avestimehr14} for two-hop erasure broadcast channels where the source does not have channel state information, and is optimal in terms of sum rates in certain regimes.
The two-user erasure broadcast channel in which only one of the receivers feeds the output back to the transmitter was studied in \cite{HeYang2017}.

\subsection{Our Contribution}

In this work, we derive a new outer bound on the capacity region of broadcast traffic in multiple input broadcast packet erasure channels with feedback, and  extend this outer bound to packet erasure relay networks  with feedback. The new outer bound combines  the standard cut-set bound technique with the capacity region of the degraded broadcast channel. We show the tightness of the outer bound for certain classes of networks. An important engineering implication of this work is that for linear network coding schemes for parallel broadcast channels, the ``xor'' packets should be sent over the correlated broadcast subchannels.

\subsection{Organization}

We describe the system model in Section \ref{sec:model}, present the new outer bounds for multiple input packet erasure channels with feedback and packet erasure relay networks with feedback in Section \ref{sec:bound1} and Section \ref{sec:bound2}, respectively. We show the tightness of the outer bounds for certain classes of networks in Section \ref{sec:example}. Section \ref{sec:conclusion} concludes this paper.

\section{System Model} \label{sec:model}
\subsection{Wireless Broadcast Packet Erasure Channel}

As introduced in Section \ref{sec:intro}, the broadcast packet erasure channel captures the broadcast nature of wireless communication by modeling each directed link between two nodes as a packet erasure channel. During each time slot, node $i$ sends out a packet of fixed size to all the nodes it is connected to. A node $j$ connected to node $i$ will receive the whole packet correctly with probability $(1-\e_{ij})$, and receive nothing with probability $\e_{ij}$. In the latter case the packet is said to be ``erased'' by the channel, and $\e_{ij}$ is called the erasure probability of link $ij$. Equivalently, in each time slot a node broadcasts the same symbol from a large field $GF(q)$ to all the nodes which it is connected to. In addition, we assume the channel is memoryless and time invariant, and erasure events over different links are independent.

More formally, consider a $1$-to-$K$ broadcast packet erasure channels, where a node $s$ is connected to $K$ nodes $t_1,t_2,\dots,t_K$. Let $[K]$ denote  $\{1,2,\dots,K\}$.  For $ j \in [K]$, the channel between nodes $s$ and $t_j$ is a packet erasure channel with erasure probability $\e_{j}$. During the $n$-th channel use, $s$ sends out a packet $X[n] \in GF(q)$.  Let $Y_j[n]$ be the symbol received by node $t_j$. If channel erasure events are independent for all links, then for any subset $A \subset [K]$,
\begin{align*}
  &\text{Prob} ( Y_j[n] = X[n], Y_{j'}[n] = *, \forall j \in A, j'\in A^C ) \nonumber \\
   = &\prod_{j\in A} ( 1 - \e_j) \prod_{j'\in A^C} \e_{j'},
\end{align*}
 where $*$ denotes that the packet has been erased by the channel, and $A^C$ denotes the complement  of $A$ in $[K]$.

A natural extension of the broadcast packet erasure channel is the broadcast packet erasure channel with feedback. In this channel we assume the transmitter immediately receives a perfect feedback from the receivers after each transmission, which indicates whether each receiver has received the packet or not. A practical mechanism for this type of feedback is the Automatic Repeat-reQuest (ARQ) protocol.

\subsection{Wireless Broadcast Packet Erasure Relay Networks}\label{Sec:PEC-RelayNet}

In wireless broadcast packet erasure relay networks, nodes are connected to other nodes via broadcast packet erasure channels. We can model a wireless broadcast erasure relay network by a directed graph $\G = (\V, \E)$, where $\V$ denotes the set of nodes and $\E$ denotes the set of links. A directed edge $(i,j) \in \E $ if node $i$ is connected to $j$, and the corresponding packet erasure probability is denoted as $\e_{ij}$. Wireless broadcast packet erasure relay networks capture the broadcast nature of wireless communication by forcing a node to send the same symbol to all the nodes it is connected to during each channel use, while we assume each node can receive packets sent from different nodes simultaneously without interference. This assumption is valid in practical systems when orthogonal schemes such as frequency-division multiplexing are implemented to avoid interference.

In this work, we consider the broadcast traffic in wireless broadcast packet erasure relay networks with channel output feedback, where there is a single source $s$, which wants to send $K$ independent messages to $K$ different destinations $t_1,t_2,\dots,t_K$ through a network of relays. Let $(R_1,R_2,\dots,R_K)$ denote the tuple of reliable transmission rates from $s$ to the $K$ destinations (the definition is the same as in \cite{DH05}). Our goal is to characterize the capacity region and the sum capacity of the network. Fig.~\ref{fig:general_network} shows an example of a three-layer $1$-to-$K$ broadcast packet erasure relay network.

\begin{figure}[t]
\centering
\includegraphics[width=0.4\textwidth]{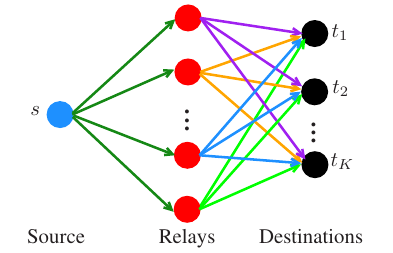}
\caption{A three-layer $1$-to-$K$ broadcast packet erasure relay network.}
\label{fig:general_network}
\end{figure}

\section{Outer Bound For Multiple Input PECs with Feedback}\label{sec:bound1}
In this section we  derive an outer bound on the capacity region of broadcast traffic in multiple input broadcast packet erasure channels with feedback.

Consider an $(M,K)$ multiple input broadcast packet erasure channels with feedback, where there is a single source $s$ connected to $K$ destinations via $M$ parallel $1$-to-$K$  broadcast packet erasure subchannels with feedback. Throughout this paper, we assume the channel is memory-less and erasure events over different subchannels are independent.

Define $[K] \triangleq \{1,2,3,\dots, K\}$ and let $\pi$ be a permutation function from $[K]$ to $[K]$. Let $\e_{ij}$ denote the channel erasure probability on the link connecting $s$ and $t_j$ of the $i$th subchannel. For any subset $A \subset [K]$, let $\e_{iA}$ denote the probability that for any $j\in A$, $t_j$ does not receive the packet on the $i$th subchannel. For example, if erasure events on all links are independent, then $\e_{iA} = \prod_{j \in A} \e_{ij}$. Lastly, we define $\pi(A) = \{ \pi(j) | j\in A\}$, for any subset $A \subset [K]$.

\begin{theorem}\label{thm:mpecouter}
For any achievable rate tuple $\R \triangleq (R_1,R_2,\dots,R_K)$, it must hold that for any permutation function $\pi: [K] \rightarrow [K]$,
\begin{align*}
\R \in \textbf{C}_{\pi},
\end{align*}
where
\begin{align*}
\textbf{C}_{\pi} \triangleq \bigl\{ &(R_1,R_2,\dots,R_K) |  \sum_{j=1}^{K}  \frac{R_{i\pi(j)}} { 1 - \e_{i\pi([j])}} \le 1,  \\
 &R_{k} = \sum_{m=1}^{M} R_{mk}, R_{i\pi(k)} \ge 0, \forall i \in [M], k \in [K] .\bigr\}
\end{align*}
\end{theorem}

Theorem \ref{thm:mpecouter} gives a natural outer bound of the capacity region, and the proof idea is the same as the proof of Proposition 1 in \cite{CCWangIT2012} by introducing auxiliary pipes connecting destinations to create physically degraded subchannels. 

\begin{IEEEproof}
For any permutation function $\pi$, we construct new multiple input broadcast erasure channels with feedback by creating information pipes connecting node $t_{\pi(j)}$ to node $t_{\pi(j+1)}$, so that $t_{\pi(j+1)}$ will get all packets node $t_{\pi(j)}$ receives,  for all $j \in [K-1]$. Therefore,  the probability that node $t_{\pi(j)}$ receives the packet sent by the transmitter on the $m$th subchannel is $1 - \e_{m\pi([j])}$. The new multiple input broadcast erasure channels are physically degraded, so feedback does not change the capacity region \cite{Gamal78}. Therefore, we can assume there is no feedback. Since the capacity region of the broadcast channel without feedback only depends on the marginal distribution of the channel output, we can further assume that the erasure events over all links in the new model are independent. Lastly, applying Theorem 3 of \cite{DH05} on the capacity region of multiple input broadcast erasure channels without feedback, we conclude that $C_{\pi}$ is the capacity region of the new model. Therefore, for any achievable rates tuple $\R$ on the original channel model, we have $\R \in C_{\pi}$.
\end{IEEEproof}

\section{Outer Bound For Broadcast Packet Erasure Relay Networks}\label{sec:bound2}
In this section we give a new outer bound on the capacity region of broadcast traffic in broadcast packet erasure relay networks with feedback. Although the construction of the new outer bound appears complicated, the idea is relatively simple. We use the standard cut-set bound technique by dividing the network of nodes into two parts. While the cut-set bound assumes nodes in each part can fully cooperate, to derive the new outer bound, we  allow nodes in the source part to fully cooperate, but allow nodes in the destination part to only partially cooperate. In this way, we obtain a multiple input broadcast PEC, the capacity region of which  upper bounds the capacity region of the original packet erasure relay networks. By applying Theorem \ref{thm:mpecouter} to the new multiple input broadcast PEC, we get the new outer bound for the original network.

Consider a wireless packet erasure relay network modeled as a directed graph $\G = (\V, \E)$, with a single source $s$ and $K$ destinations $t_1, t_2, \dots,t_K$. As described in Section~\ref{Sec:PEC-RelayNet}, each directed edge $(i,j) \in \E $ has a associated packet erasure probability of $\e_{ij}$.  Let $A \subset \V$, which does not contain the source node $s$, i.e., $s \notin A$. Let $J(A) = \{ j | t_j \in A \}$, the set of indices of the destinations contained in $A$. Let $\E_A$ denote the edge cut corresponding to $A$. More precisely,
\begin{align*}
\E_A \triangleq \{(v,w) | (v,w)\in \E, v\notin A, w\in A \}.
\end{align*}

Let $\V_A \triangleq \{v\in \V | \exists w \in A, s.t.  (v,w) \in \E_A \}$, i.e., $\V_A$ denotes the set of nodes which are in $ \V / A$ and form the edge cut.  For any subset $\W_A \subset \V_A$,
we will derive an upper bound on the achievable rates tuple $ R_{J(A)} \triangleq \{R_i\}_{ i\in J(A)}$ in terms of $\E_A$ and $\W_A$.

We use the following algorithm to construct an edge set $\E^*$ and a vertex  set $\V^*$.

\begin{enumerate}
\item Initialization: set $\E^* = \emptyset, \V^* = \emptyset $.

\item For each node $w \in A$, if all edges which go to $w$ come from nodes in $\W_A$, then a) add these edges to $\E^*$, b) if $w$ is one of the $K$ destinations, add $w$ to $\V^*$, c) delete all edges that go to $w$ and leave from $w$, and delete $w$.

\item  For each node $w \in A$, if  no edges go to $w$, then a) delete all edges that leave from $w$, b) add $w$ to $\V^*$ if $w$ is one of the $K$ destinations, c) delete $w$.


\item Repeat step 2 and step 3 until no node will be deleted.

\item Consider the subgraph consisting of only nodes in $A$ which have not been deleted. Make the subgraph undirected, and then find out all the connected components. Delete the connected components which does not contain any destinations. Denote the remaining connected components by $\C_1, \C_2, \dots, \C_p$, where $\C_i$ is the set of vertices in the corresponding components. Then merge each component to a single super node, denoted by $v_1,v_2, \dots, v_p$.

\end{enumerate}

Now we use $\E_A, \V_A, \W_A, \E^*, \V^* $ to construct  multiple input broadcast PECs. The capacity region of this new network will yield an outer bound on the rate tuples $R_{J(A)}$.

For each node $v$ in $\V_A$, we process $v$ as follows:
\begin{enumerate}

\item Consider the edges which leave from $v$ and are in $\E_A$, and then divide them into two subsets $\E_1$ and $\E_2$: $\E_1$ is the set of edges in $\E^*$, and $\E_2$ is the set of edges in $\E_A \backslash  \E^*$.

\item Split node $v$ into two nodes $v_a$ and $v_b$.

\item Connect $v_a$ to each node in $\V^*$ and the super nodes $v_1,v_2,\dots,v_p$ via broadcast  PECs with feedback, with erasure probability $ \prod_{(i,j)\in \E_1} \e_{ij}$. If $\E_1$ is an empty set,  set the erasure probability to be $1$ (or equivalently,  we can delete node $v_a$). Further, we assume the erasure events over all these links are  the same, i.e., if the packet over some link is erased, then packets over all links are erased.

\item Connect $v_b$ to each super node $v_1,v_2,\dots,v_p$ via a broadcast packet erasure channel with feedback. The erasure probability of each link is set as follows. For each $i \in [p]$, let $\hat{\E}_i$ denote the set of edges  in  $\E_2$  which leave from $v$ and go to some node in $\C_i$. If $\E_i$ is an empty set, then set the erasure probability to be $1$; otherwise, set the erasure probability to be $\prod_{(i,j)\in \hat{\E}_i} \e_{ij}$.

\end{enumerate}

After processing  all the nodes $v$ in $\V_A$, we allow each  newly created node to fully cooperate with each other and the source node $s$. Let $N$ be the total  number of nodes in $\V^*$ plus $p$.  Therefore we get a multiple input $1$-to-$N$ broadcast packet erasure channel with feedback. Note that, by construction, the erasure events over all links in the subchannel  are not independent for some subchannels. (Indeed they are completely correlated.)

Next we upper bound $R_{J(A)}$ by using an outer bound on the capacity region of this multiple input broadcast PEC with feedback. By our construction, for each $j\in J(A)$, the destination $t_j$ is  contained in either  $\V^*$ or $\C_i$ for some $i\in [p]$. For each $ i \in [p]$, let $J_i$ denote the set of indices of destinations which are in $\C_i$. More precisely, $J_i \triangleq \{j\in J(A) | t_j \in C_i \}$.

Relabel the $N$ sink nodes in the new multiple input PECs by $d_1,d_2,\dots, d_N$.   Define $N$ auxillary variables $Q_1, Q_2, \dots, Q_N$, where $Q_i \triangleq R_j$ if $d_i$ is the single destination $t_j$, or $Q_i \triangleq \sum_{k\in J_j} R_k$ if $d_i$ is the super node $v_j$.

We can use Theorem \ref{thm:mpecouter} to upper bound the capacity region of this new $1$-to-$N$ multiple input PEC with feedback. Let $\N(A, \W_A)$ denote this new multiple input broadcast PEC with feedback, and let $R_{outer} (\N(A, \W_A))$ denote the outer bound of the capacity region of $\N(A, \W_A)$ given in Theorem \ref{thm:mpecouter}. Then
\begin{theorem} \label{thm:uppnetwork}
\begin{align}
  (Q_1,Q_2,\dots,Q_N)  \in R_{outer} (\N(A, \W_A)). \label{eqn:uppQ}
\end{align}
  
\end{theorem}
Therefore, Equation \eqref{eqn:uppQ} gives an outer bound on $R_{J(A)}$.

\begin{proof}
 We only need to argue that the new channel $\N(A, \W_A)$ is no worse than the original network for destinations $d_j, j\in J(A)$. First note that, there is no loss by deleting nodes of which the incoming edges come from $\W_A$, since in the last step we connect these edges directly  to all destinations. Second, there is no loss by merging the connected component to a super node, since in this way all the nodes in the connected component are assumed to fully cooperate with each other and behave like a single node.  Third, there is no loss by assuming  the new created nodes can cooperate and share the same message with source $s$. Therefore, the capacity region of the new constructed multiple input packet erasure channel upper bounds the achievable rates tuple $R_J(A)$.
\end{proof}

\section{Tightness of Outerbounds}\label{sec:example}
In this section, we show that for certain classes of networks, the outer bounds derived in Section \ref{sec:bound1} and Section \ref{sec:bound2}  on the capacity region of multiple input packet erasure channels with feedback and broadcast packet erasure relay networks with feedback are tight in terms of sum rate.

Theorem \ref{thm:mpecouter} gives an outer bound on the capacity region of multiple input broadcast erasure channels with feedback. A natural inner bound can be derived  by adding the capacity region of each subchannel. In general, this inner bound without coding across subchannels does not match the outer bound. In the following, we show that for certain simple two-input $1$-to-$2$ broadcast packet erasure channels with feedback, while the above inner bound does not match the outer bound, the maximum sum rate of the outer bound can be achieved by coding across the two subchannels.

\begin{figure}[t]
\centering
\includegraphics[width=0.40\textwidth]{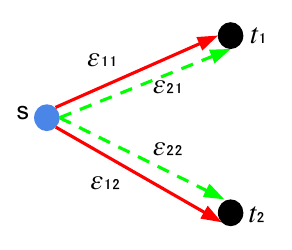}
\caption{A two-input $1$-to-$2$ broadcast packet erasure channels with feedback.}
\label{fig:one_to_two}
\end{figure}

Consider a two-input $1$-to-$2$ broadcast packet erasure channels with feedback where there are two destinations $t_1$ and $t_2$ and two subchannels (see Fig.~\ref{fig:one_to_two}). In the first subchannel, the packet erasure probabilities are $\e_{11} = \e_{12} \triangleq \e_1$ and erasure events on the two links are independent of each other. In the second subchannel, $\e_{21} = \e_{22} \triangleq \e_2$ and erasure events on the two links are the same, i.e., at any time, $t_1$ receives a packet if and only if $t_2$ also receives the packet. We assume erasure events on different subchannels are independent.

\begin{theorem}\label{thm:ex1}
  If $\e_2 \ge 1 - \frac{(1-\e_1)\e_1}{2}$ and $0 < \e_1,\e_2 < 1$, then the maximum sum rate of the outer bound in Theorem \ref{thm:mpecouter} is tight and can be  achieved by coding across subchannels, which thus characterizes the sum capacity of this channel. In addition, the inner bound without coding across subchannels is strictly sub-optimal.
\end{theorem}

\begin{IEEEproof}
First, we calculate the maximum sum rate of the inner bound via independent communication scheme over the two subchannels.

For the second subchannel, since the erasure events on the two links are the same, the sum rates $1-\e_2$ can be achievable via time-sharing over the second subchannel.

Applying the result from \cite{GT09}, the capacity region of the first subchannel is the intersection of 
\begin{align}
\{ (R_{11}, R_{12}) \; | \;  \frac{R_{11}}{1-\e_1} +  \frac{R_{12}}{1-\e_1^2} \le 1, R_{11}, R_{12} \ge 0 \}
\end{align}
and 
\begin{align}
\{ (R_{11}, R_{12}) \; | \; \frac{R_{11}}{1-\e_1^2} +  \frac{R_{12}}{1-\e_1} \le 1, R_{11}, R_{12} \ge 0 \}.
\end{align}

Due to the symmetry and convexity of the capacity region, the sum rates over the first subchannel is maximized when $R_{11}=R_{12}$. Set $R_{11}  = R_{12} = a$. Then we have 
\[  \frac{a}{1-\e_1} +  \frac{a}{1-\e_1^2} \le 1, \]
 and thus $a \le \dfrac{1-\e_1^2}{2+\e_1}$.

Therefore, the maximum sum rates over the first subchannel is $\frac{2(1-\e_1^2)}{2 + \e_1}$. Combining the sum rates from both subchannels, we have

\begin{align*}
 R^{inner}_{sum} = \frac{2(1-\e_1^2)}{2 + \e_1} + 1 - \e_2.
 \end{align*}
 
 

Next, we derive an outerbound on the sum rates using Theorem  \ref{thm:mpecouter}.
 By Theorem \ref{thm:mpecouter}, the outer bound on the capacity region is
  $\V \triangleq \V_1 \cap \V_2$,
where 
\begin{align}
\V_1 \triangleq  \{(R_1,R_2) | & \frac{R_{11}}{1 - \e_1^2} + \frac{R_{12}}{1 - \e_1} \le 1, \label{eqn:t1} \\
& R_{21} + R_{22} \le 1 - \e_2, \label{eqn:t2} \\
& R_1 = R_{11} + R_{21}, \label{eqn:t3} \\
& R_2 = R_{12} + R_{22}, \label{eqn:t4} \\
&R_{ij} \ge 0, \forall i,j \in [2]\}, \nonumber
\end{align}
 and
 \begin{align*}
 \V_2 \triangleq \{(R_1,R_2) | & \frac{R_{11}}{1 - \e_1} + \frac{R_{12}}{1 - \e_1^2} \le 1, \\
 & R_{21} + R_{22} \le 1 - \e_2, \\
 & R_1 = R_{11} + R_{21}, \\
 & R_2 = R_{12} + R_{22}, \\
 & R_{ij} \ge 0, \forall i,j \in [2]\}.
 \end{align*}

Suppose $(R_1^*, R_2^*)$ is the solution to
\begin{align*} 
\max_{(R_1,R_2) \in  \V} R_1 + R_2.
\end{align*}
Due to the symmetry and convexity of $\V$, without loss of generality, we can assume $R_1^* = R_2^*$. Therefore,
\begin{align*}
  \max_{(R_1,R_2) \in  \V} R_1 + R_2 = \max_{(R_1, R_2) \in \V_1, R_1 = R_2}  R_1 + R_2.
\end{align*}
Let $R_1 = R_2 = x$. Then from \eqref{eqn:t1}, \eqref{eqn:t3} and \eqref{eqn:t4}, we have
\begin{align*}
  \frac{x - R_{21}}{1 - \e_1^2} + \frac{x - R_{22}}{1 - \e_1} \le 1.
\end{align*}
Therefore, $ x \frac{2 + \e_1}{ 1 - \e_1^2} \le 1 + \frac{R_{21}}{ 1-\e_1^2} + \frac{R_{22}}{1 - \e_1}$. Due to \eqref{eqn:t2}, $R_{21} + R_{22} \le 1 - \e_2$, and thus
\begin{align*}
x \frac{2 + \e_1}{ 1 - \e_1^2} &\le 1 + \frac{R_{21}}{ 1-\e_1^2} + \frac{R_{22}}{1 - \e_1}  \le 1 + \frac{R_{21}}{ 1-\e_1} + \frac{R_{22}}{1 - \e_1} \\
&= 1 +  \frac{R_{21} + R_{22} }{ 1-\e_1}  \le 1 + \frac{1-\e_2}{1-\e_1}.
\end{align*}
Therefore, $x \le \dfrac{(2 - \e_1 - \e_2)(1+\e_1)}{2 + \e_1}$.

When $x = \dfrac{(2 - \e_1 - \e_2)(1+\e_1)}{2 + \e_1}$,
\begin{align*}
  R_{11} &= x \ge 0, R_{21} = 0 ,\\
  R_{12} &= x - (1 - \e_2) = \frac{\e_2 - \e_1^2}{ 2 + \e_1} \ge \frac{2 - \e_1 - \e_1^2}{2(2 + \e_1)} \ge 0,\\
  R_{22} &= 1 - \e_2 \ge 0,   
\end{align*}
so this is a valid solution, i.e., $(x,x) \in \V_1 \cap \V_2 $. Therefore, the maximum sum rates of the outer bound is
 \begin{align}\label{eqn:uppsum}
  R^{outer}_{sum} = \dfrac{2(2 - \e_1 - \e_2)(1+\e_1)}{2 + \e_1}.
 \end{align}

It is easy to verify that there is a nonzero gap between the inner bound and outer bound:
\begin{align*}
   & R^{outer}_{sum} - R^{inner}_{sum} \nonumber \\
=& \frac{2(2 - \e_1 - \e_2)(1+\e_1)}{2 + \e_1} - (\frac{2(1-\e_1^2)}{2 + \e_1} + 1 - \e_2 ) \\
=& \frac{\e_1 (1-\e_2)}{2 + \e_1} > 0,
\end{align*}
for $0 < \e_1,\e_2 <1$.

Next, we describe a coding scheme to achieve the outer bound $R^{outer}_{sum}$. The scheme is essentially the same as the two-phase scheme introduced in \cite{GT09}, except that we use the second subchannel only to send the ``xor'' packets. For the first subchannel:
\begin{enumerate}
  \item In the first $N$ transmissions, the source $s$ sends packets for destination $t_1$. The source resends the same packet if and only if neither  $t_1$ nor $t_2$ receives the packet. After the first $N$ transmissions, on average, $t_1$ receives $N(1 - \e_1)$ packets, and $t_2$ receives $N(1-\e_1)\e_1$ packets which are for $t_1$ but not received by $t_1$. Denote these packets received by $t_2$ only by $P_1$.

  \item Similarly, during the second $N$ transmissions, the source $s$ sends packets for destination $t_2$. The source resends the same packet if and only if neither $t_1$ nor $t_2$ receives the packet. After the second $N$ transmissions, on average, $t_2$ receives $N(1 - \e_1)$ packets, and $t_1$ receives $N(1-\e_1)\e_1$ packets which are for $t_2$ but not received by $t_2$. Denote these packets received by $t_1$ only by $P_2$.

  \item In the following $N^*$ transmissions, where $N^*$ is to be determined soon, $s$ sends out packets which are random linear combinations of $P_1$ and $P_2$.
\end{enumerate}

For the second subchannel, during the $N + N + N^*$ transmissions, $s$ only sends out random linear combinations of $P_1$ and $P_2$. Note that one can use block coding to resolve the noncausality issue.

The destination $t_1$ needs $N(1 - \e_1)\e_1$ packets  which are random linear combinations of $P_1$ and $P_2$ to decode $P_1$. Similarly, the destination $t_2$ needs $N(1 - \e_1)\e_1$ packets  which are random linear combinations of $P_1$ and $P_2$ to decode $P_2$. After all the transmissions, on average, both $t_1$ and $t_2$ will receive
\begin{align*}
  N^*(1-\e_1) + (2 N + N^*) (1 - \e_2)
\end{align*}
random linear combinations of $P_1$ and $P_2$.

Therefore, $N^*(1-\e_1) + (2 N + N^*) (1 - \e_2) = N (1 - \e_1)\e_1$, and thus 
\begin{align*}
N^* = \frac{(1-\e_1)\e_1 - 2(1-\e_2)}{2 - \e_1 - \e_2}N,
\end{align*}
which is positive when $\e_2 \ge 1 - \frac{(1-\e_1)\e_1}{2}$.

After these $N+N+N^*$ transmissions, $t_1$ can decode $N(1-\e_1) + N(1-\e_1)\e_1 = N(1-\e_1^2)$ packets for it, and similarly  $t_2$ can also decode $N(1-\e_1) + N(1-\e_1)\e_1 = N(1-\e_1^2)$ packets for it. Therefore, the achieved sum rate is
\begin{align*}
  \frac{N(1-\e_1^2) + N(1-\e_1^2) }{ N + N + N^*} &= \frac{2(1-\e_1^2)}{2 + \frac{(1-\e_1)\e_1 - 2(1-\e_2)}{2 - \e_1 - \e_2} } \\
& = \frac{2 (1 + \e_1)( 2 - \e_1 - \e_2)}{2 + \e_1} \\
&= R^{outer}_{sum}.
\end{align*}

Therefore this coding scheme achieves $R^{outer}_{sum}$, which is thus the sum capacity of this channel.

\end{IEEEproof}

Based on the above two input $1$-to-$2$ broadcast PECs with feedback, we can construct a  packet erasure relay network with feedback where the new outer bound in Theorem \ref{thm:uppnetwork} is tight in terms of sum rate.

\begin{figure}[t]
\centering
\includegraphics[width=0.4\textwidth]{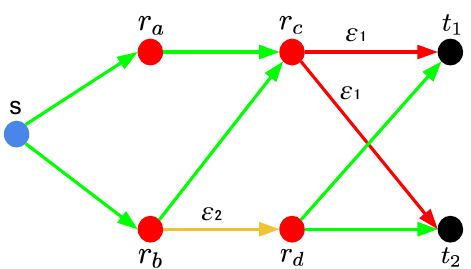}
\caption{A $1$-to-$2$ broadcast packet erasure relay network.}
\label{fig:network}
\end{figure}

Consider the broadcast packet erasure relay network in Fig.~\ref{fig:network}, where all links are independent and the packet erasure probabilities are zero for all links except the links $(r_c, t_1), (r_c, t_2)$ and $(r_b, r_d)$ with erasure probability $\e_1,\e_1$ and $\e_2$, respectively. In addition, $\e_1$ and $\e_2$ satisfy the assumption in Theorem \ref{thm:ex1}.
We can derive the sum capacity of this network by applying Theorem \ref{thm:uppnetwork} as follows. In the algorithm given in Section \ref{sec:bound2}, take $A = \{t_1, t_2, r_d\}$. Thus $\E_A = \{(r_c, t_1), (r_c, t_2), (r_b, r_d) \}$ and $\V_A = \{r_b, r_c \}$. Take $\W_A = \{ r_b\}$. Then the new network $\N(A, \W_A)$ constructed by the algorithm is exactly the channel described in Theorem \ref{thm:ex1} with the same channel parameters. Therefore, Theorem \ref{thm:uppnetwork} upper bounds the sum rate of the network in Fig.~\ref{fig:network} by $R^{outer}_{sum}$ in Equation \eqref{eqn:uppsum}. Furthermore,  it is  achievable by using the corresponding scheme in the proof of Theorem \ref{thm:ex1}.

\section{Conclusion}\label{sec:conclusion}

We derive a new outer bound on the capacity region of broadcast traffic in multiple input broadcast packet erasure channels with feedback, and  extend this outer bound to packet erasure relay networks  with feedback. The new outer bound involves  the standard cut-set bound technique and the capacity region of the degraded broadcast channel. We show the tightness of the outer bound for certain classes of networks. One important engineering implication of this work is that for network coding schemes for parallel broadcast channels, the ``xor'' packets should be sent over the correlated broadcast subchannels.

\section*{Acknowledgment}\label{sec:ackowledgementd}

This work was done when Quan Geng was with University of Illinois.
This work was done when Hieu Do was with KTH Royal Institute of Technology in Stockholm, and was visiting University of Illinois. 

The authors would like to thank Dr. Pramod Viswanath and Dr. Sreeram Kannan
 for the helpful discussions.

\bibliographystyle{IEEEtran}
\bibliography{reference}

\end{document}